\documentclass[reqno]{amsproc}

\usepackage{amsmath}
\usepackage{enumerate}
\usepackage{hyperref}
\usepackage{tikz}
\tikzstyle{infinito}=[circle,inner sep=0pt,minimum size=0mm]
\tikzstyle{nodo}=[circle,draw,fill,inner sep=0pt,minimum size=2.5pt]

\newtheorem{theorem}{Theorem}
\newtheorem{lemma}{Lemma}[section]
\newtheorem{proposition}[lemma]{Proposition}

\theoremstyle{definition}
\newtheorem{definition}[lemma]{Definition}
\newcounter{assumptionalpha}

\newtheorem{assumption}[assumptionalpha]{Assumption}

\theoremstyle{remark}

\numberwithin{equation}{section}

\newcommand{\CO}{\mathbb{C}}
\newcommand{\NA}{\mathbb{N}}
\newcommand{\RE}{\mathbb{R}}
\newcommand{\GG}{\mathcal{G}}
\newcommand{\KK}{\mathcal{K}}
\newcommand{\Ver}{\mathcal{V}}
\newcommand{\Ed}{\mathcal{E}}

\newcommand{\ve}{\varepsilon}
\newcommand{\supp}{\operatorname{supp}}
\newcommand{\sech}{\operatorname{sech}}

\newcommand{\VV}{R}
\newcommand{\WW}{S}

\newcommand{\x}{x}

\begin{document}

\title[Existence of the ground state for the NLS with potential on graphs]{Existence of the ground state for the NLS with potential on graphs}

\author{Claudio Cacciapuoti}
\address{DiSAT, Sezione di Matematica, Universit\`a dell'Insubria, via Valleggio 11, I-22100 Como, Italy}
\email{claudio.cacciapuoti@uninsubria.it}
\thanks{The author  acknowledges the support of the FIR 2013 project ``Condensed Matter in Mathematical Physics'', Ministry of University and
 Research of Italian Republic  (code RBFR13WAET)}


\date{}

\begin{abstract}We review and extend several recent results on the existence of  the ground state for the nonlinear Schr\"odinger (NLS) equation on a metric graph. By ground state we mean a minimizer  of  the NLS energy functional constrained to the manifold of fixed $L^2$-norm. In the energy functional we allow for the presence of a potential term, of delta-interactions in the vertices of the graph, and of a power-type focusing nonlinear term. We discuss both subcritical and critical nonlinearity. Under general assumptions on the graph and the potential, we prove that a ground state exists for sufficiently small mass, whenever the constrained infimum of the quadratic part of the energy functional  is strictly negative. 
\end{abstract}

\maketitle

\section{Introduction}
Analysis on metric graphs and networks is a very well established research field,  potentially  with many  physical and technological  applications. From a mathematical point of view the interest in these structures lies in the fact that, despite being  essentially simple  one-dimensional  objects, they still exhibit several intriguing  features due to nontrivial connectivity and topology. 

For an introduction to metric graphs and an extended list of references we refer to one of the many monographs on the subject, see, e.g., \cite{berkolaiko-kuchment13,exner-keating-kuchment-sunada-teplyaev08,mugnolo14,post12}. 

The study  of nonlinear equations on graphs  is still at its beginning yet quickly developing. A  monograph on quasilinear wave equations on one-dimensional networks, mostly dealing with the problem of  the well-posedness,   is \cite{mehmeti94} (see also \cite{bona-cascaval-cam08,mehmeti-vonbelow-nicaise01}).

Concerning the NLS equation on simple networks (e.g., the $Y$-junction or  star-graph, see Fig. \ref{f:2}) a certain amount of work has been recently carried on: for the scattering and transmission properties of simple networks, see, e.g.,   \cite{adami-cacciapuoti-finco-noja-rmp11,cascaval-hunter-lm10,sobirov-matrasulov-sabirov-sawada-nakamura-pre10,sobirov-sabirov-matrasulov-saidov-nakamura-npcn14};  the inverse scattering method has recently  been applied to the cubic NLS on a star-graph in  \cite{caudrelier-cmp15};  the shrinking limit for the dynamics in a  thin network (a relevant problem from the point of view of applications)   has been studied in \cite{kevrekidis-frantzeskakis-theocharis-kevrekidis-pla03,sobirov-babajanov-matrasulov-npcm17,uecker-grieser-sobirov-babajanov-matrasulov-pre15}.  For a review on recent results and open problems related to the NLS equation on graphs  we refer to \cite{noja-ptrsa14}.

In what follows we shall  focus attention on the problem of the existence of the ground state. We shall discuss several related works at the end of the introduction. \\

The problem we are interested in is the minimization of the nonlinear Schr\"odinger energy functional 
\begin{equation}\label{energy}
E[\Psi]:= \| \Psi ' \|^2 +(\Psi,W\Psi) +   \sum_{v\in \Ver} \alpha_v |\Psi(v)|^2 -\frac{1}{\mu+1} \|\Psi \|_{2\mu+2}^{2\mu + 2}  \qquad 0<\mu \leq 2
\end{equation}
defined on a metric graph $\GG$, where $W$ is a potential on the graph, $\Ver$ is the set of vertices of the graph, and $\alpha_v$ are some real constants that take into account  possible delta-interactions in the vertices. We shall focus attention on the minimization problem 
\begin{equation}\label{minimization}
-\nu_\mu(m) := \inf \{ E[\Psi]\,|\;  \Psi\in H^1(\GG), \, \|\Psi\|^2= m \}.
\end{equation}
The parameter $m$ is called \emph{mass}, and will play an important role in our analysis.  It is easy to check that the functional $E[\Psi]$, defined on the space $H^1(\GG)$ (see Sec. \ref{s:prel} for a precise definition of $H^1(\GG)$) is unbounded from below; to this aim it is enough to consider the behavior of $E[\lambda\Psi]$ for large $\lambda$. It is well known that imposing the constraint $\|\Psi\|^2 = m$ may solve this issue; indeed, the first question that we try to answer concerns the existence of a lower bound for the infimum in Eq. \eqref{minimization}. 

Whenever such a lower bound exists, we will be concerned about the existence of a minimizer. We shall use the following definition
\begin{definition}[Ground state]
A minimizer of problem \eqref{minimization}, i.e., a function $\hat \Psi \in H^1(\GG)$, such that $\|\hat \Psi\|^2=m$, and $E[\hat \Psi] = -\nu_\mu(m)$ (if it exists), is  called \emph{ground state} (of mass $m$).
\end{definition}
We recall that a metric graph can be understood as a metric space made up of a set of segments, referred  to as \emph{edges}, and a set of points, called \emph{vertices}. 
 In our analysis,  we shall always assume that the graph has a finite number of edges and vertices.  The edges can be of  finite or infinite length; in the first case each edge  is identified with a segment $[0,\ell_e]$ ($\ell_e$ being the length of the  edge), in the latter  with a copy of the half-line $[0,+\infty)$. The edges of the graph are glued together according to a connection map  which identifies each endpoint ($0$ or $\ell_e$) of each edge with a vertex $v$ of the graph, see Fig. \ref{f:1}.

As for the standard NLS with power-type nonlinearity on the real-line, the case $0<\mu<2$ is called \emph{subcritical} while the case $\mu=2$ is called \emph{critical}.   The terminology is associated to the scaling properties of the \emph{kinetic} and \emph{nonquadratic}  terms in the energy  functional ($\| \Psi ' \|^2$ and $\|\Psi \|^{2\mu+2}_{2\mu + 2} $ respectively). In fact, under the mass invariant transformation $\Psi \in L^2(\GG) \to \Psi_\lambda \in L^2(\lambda^{-1} \GG)$, defined by $\Psi_\lambda(x) := \sqrt\lambda \Psi(\lambda x)$,  the scaling relations are  
\begin{equation}\label{scaling}
\|\Psi_\lambda'\|_{L^2(\lambda^{-1}\GG)}^2 = \lambda^{2}\|\Psi'\|_{L^2(\GG)}^2 \quad \text{and} \quad  \|\Psi_\lambda\|_{L^{2\mu+2}(\lambda^{-1}\GG)}^{2\mu+2} = \lambda^{\mu}\|\Psi\|_{L^{2\mu+2}(\GG)}^{2\mu+2}.
\end{equation} 
When $\GG$ is itself scale invariant, i.e., $\GG$ coincides with the real-line, or the half-line, or $N$ half lines with a common vertex (a \emph{star-graph}), the scaling \eqref{scaling} clearly implies that the kinetic (positive) term dominates for $0<\mu<2$ and $\lambda$ large enough, thus suggesting that the infimum \eqref{minimization} is bounded from below. For  $\mu=2$  the two terms scale in the same way, and it turns out that the infimum $-\nu_\mu(m)$ is bounded from below only for small mass. The same behavior can be observed for generic graphs. 

Our main results are stated in Ths. \ref{t:lowerbound} and \ref{t:main} below.  We remark that the subcritical case was discussed in \cite{cacciapuoti-finco-noja-rxv16}, here we include  the case $\mu=2$. When $\GG$ is a star-graph (and $W=0$), the critical case was already discussed in \cite{adami-cacciapuoti-finco-noja-aihp14}, to extend the analysis  to generic graphs   we shall use several results from \cite{adami-serra-tilli-cmp16}. 

We make the following assumptions. 
\begin{assumption}\label{a:G}$\GG$ is a finite, connected graph, with at  least one external edge.  
\end{assumption}
We call \emph{finite} a graph that has a finite number of edges and vertices; a graph $\GG$ is \emph{connected} if given any two points of the graph there is always a path  in $\GG$ connecting them; and we call \emph{external} an edge of infinite length.
\begin{assumption}\label{a:W}
$W = W_+ - W_-$ with $W_\pm\geq 0$, $W_+\in L^1(\GG) + L^{\infty}(\GG)$, and $W_-\in L^r(\GG)$ for some $r\in [1,1+1/\mu]$.
\end{assumption}
We recall that $L^r(\GG)\subset L^1(\GG) + L^{\infty}(\GG)$ for all $r\in [1,+\infty ]$.  
Hence, under Ass. \ref{a:W},  one has $W\in L^1(\GG)+L^\infty(\GG)$. We remark that we shall use the stronger assumption $W_-\in L^r(\GG)$ only in   Th. \ref{t:main}. More precisely,  the additional constraint $r\in[1,1+1/\mu]$ is needed only   in the bound \eqref{W-1}; in  the remaining part of the proof it is enough to assume  $1\leq r <\infty$.

Denote by $E^{lin}[\Psi]$  the quadratic part of energy functional $E[\Psi]$, 
\begin{equation}\label{Elin}
E^{lin}[\Psi ] :=  \| \Psi ' \|^2 +(\Psi,W\Psi) +   \sum_{v\in \Ver} \alpha_v |\Psi(v)|^2,
\end{equation}
and by $-E_0$ the  infimum 
\begin{equation}\label{E0}
-E_0 : = \inf \left\{ E^{lin} [ \Psi] \,|\;  \Psi\in H^1(\GG),\; \|\Psi\|^2 = 1\right\}.
\end{equation}
We shall assume that this infimum is negative.  
\begin{assumption}\label{a:E0} 
$E_0>0$. 
\end{assumption}
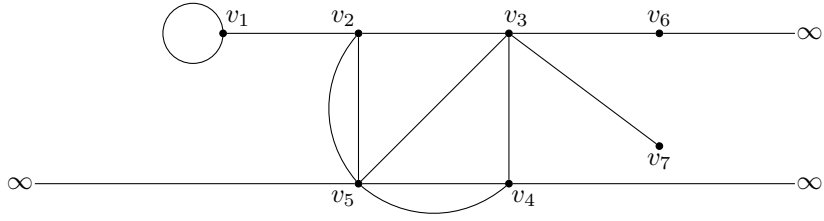
\begin{figure}[t]
\begin{tikzpicture}
\node at (-6.5,0) [infinito](0) {${\infty}$};
\node at (-2,0) [nodo] (5) {};
\node at (0,0) [nodo] (4) {};
\node at (-2,2) [nodo] (2) {};
\node at (0,2) [nodo] (3) {};
\node at (4,0) [infinito] (8) {${\infty}$};
\node at (4,2) [infinito] (7) {${\infty}$};
\node at (2,2) [nodo] (6) {};
\node at (-3.8,2) [nodo] (1) {};
\node at (-2.2,2.2){$v_2$};
\node at (0.1,2.2){$v_3$};
\node at (0.2,-0.2){$v_4$};
\node at (-2.2,-0.2){$v_5$};
\node at (2,2.2){$v_6$};
\node at (-3.6,2.2){$v_1$};
\draw [-,black] (-4.2,2) circle (0.4cm) ;
\draw [-,black] (5) -- (2);
\draw [-,black] (0) -- (5);
\draw [-,black] (5) -- (4);
\draw [-,black] (5) to [out=-40,in=-140] (4);
\draw [-,black] (5) to [out=130,in=-130] (2);
\draw [-,black] (5) -- (3);
\draw [-,black] (1) -- (2);
\draw [-,black] (2) -- (3);
\draw [-,black] (3) -- (6);
\draw [-,black] (4) -- (3);
\draw [-,black] (6) -- (7);
\draw [-,black] (4) -- (8);
\node at (2,0.5) [nodo] (10) {};
\node at (2,0.3){$v_7$};
\draw [-,black] (3) -- (10);
\end{tikzpicture}
\caption{\label{f:1}A graph with  $7$ vertices and  $14$ edges: $3$ external, and $11$ internal (of which $1$ terminal).}
\end{figure}
In the statements of our main results, some constraints on $\mu$ and $m$,  with an interplay between the two parameters,  are needed. 
Two quantities will enter the constraints. For the critical case it will be relevant the best constant $K_{6,2}(\GG)$ satisfying the Gagliardo-Nirenberg inequality \eqref{gn1}. In the subcritical case, instead, the constraint will depend on a constant, denoted by  $\gamma_\mu$, which is related to the infimum of the ``free'' nonlinear NLS energy functional on the real-line 
\[ -t_\mu(m) : = \inf \left\{E_\RE[\psi] \,\Big|\; \psi\in H^1(\RE), \, \| \psi\|_{L^2(\RE)}^2 =m  \right\} \qquad 0<\mu\leq 2,\]
with 
\[E_\RE[\psi]:= \|\psi'\|^{2}_{L^2(\RE)}-\frac{1}{\mu+1} \|\psi\|^{2\mu+2}_{L^{2\mu+2}(\RE)}. \]
It turns out that $t_\mu(m)$ can be explicitly computed, and in particular 
\begin{equation}\label{thres-inf}
 t_\mu(m) = \gamma_\mu m^{1+\frac{2\mu}{2-\mu}} \qquad\text{for}\quad  0<\mu<2, 
\end{equation}
where $\gamma_\mu$ is positive  (its explicit expression is  given in Eq. \eqref{gammamu}).
 
Our first result concerns the existence of a lower bound for the infimum in Eq. \eqref{minimization}. 
\begin{theorem}\label{t:lowerbound}
Let Assumption \ref{a:G} hold true and assume  $W\in L^1(\GG) + L^{\infty}(\GG)$. If $0<\mu<2$ then $\nu_\mu(m) <+\infty$ for any $m>0$. If $\mu = 2$ then
$\nu_\mu(m)<+\infty$ for any  $0<m<\sqrt3/K_{6,2}^3(\GG)$.
\end{theorem}

The second result concerns the existence of the ground state.  
\begin{theorem} \label{t:main}
Let Assumptions  \ref{a:G}, \ref{a:W}, and \ref{a:E0} hold true. Then $-\nu_\mu(m)\leq-mE_0$. Moreover, let 
\begin{equation}\label{m_mustar}m_\mu^* := \left\{ \begin{aligned}
& \left(E_0/\gamma_\mu\right)^{\frac1\mu-\frac12} \qquad  && \text{if} \quad0<\mu<2 \\
& \sqrt3/K_{6,2}^3(\GG) && \text{if} \quad \mu =2
\end{aligned}\right.
\end{equation}
Then the ground state $\hat\Psi$ exists for all $0<m<m_\mu^*$. 
\end{theorem}
%
%

We remark that it is possible to show that the upper bound for $-\nu_\mu(m)$ is strict, i.e., $-\nu_\mu(m)<-mE_0$. However, since we do not need this additional information, we will not pursue this goal. 

A main tool to prove Th. \ref{t:main} is the concentration-compactness lemma (Lem. \ref{l:cc} below). We remark that concentration-compactness methods are standard in $\RE^d$ (see, e.g., \cite{cazenave03}). Here we adapt the technique to a setting where there is no translation invariance. 

Once that the existence of a lower bound for the infimum \eqref{minimization} is granted by Th. \ref{t:lowerbound},  Lemma \ref{l:cc} is used to show that a sufficient condition for the existence of a minimizer is 
\begin{equation}\label{sufficient}
-\nu_\mu(m) < -t_\mu(m),
\end{equation}
both in the subcritical and in the critical case. Condition \eqref{sufficient} can be used together with the fact that $t_\mu(m)$ can be explicitly computed and that it must be $-\nu_\mu(m)\leq-E_0 m$. In the subcritical case, the threshold mass $m^*_\mu$ follows by requiring  $-E_0m < -t_\mu(m)$.  In the critical case the situation is slightly different because  $t_2(m)$ has a sharp transition depending on the value of $m$, precisely  
\begin{equation}\label{nuRE}
 -t_2(m) = \left\{ 
 \begin{aligned}
& 0 \qquad & \text{if} \; m \leq \pi \sqrt 3/2 \\ 
 &-\infty & \text{if} \; m>\pi \sqrt 3/ 2   \end{aligned}\right. 
\end{equation}
Note that  $\pi \sqrt 3/ 2 = \sqrt3/K_{6,2}^3(\RE)$, where $K_{6,2}(\RE)$ is the best Gagliardo-Nirenberg constant for the real-line. In this case the value of $m^*_2$ arises from the thresholds in Th. \ref{t:lowerbound} and Eq. \eqref{nuRE}, together with the fact that $K_{6,2}(\GG)\geq K_{6,2}(\RE)$, see \cite{adami-serra-tilli-cmp16}.

We  remark that, in the subcritical case and for $W\in L^r(\GG)$ for some $r\in[1,+\infty)$, one can prove that the infimum  $-\nu_\mu(m)$ cannot exceed $-t_\mu(m)$, see Prop. \ref{p:necessary} below. \\

We conclude the introduction with several remarks and a discussion  on the related literature. 

\subsection*{Conservation laws and well-posedness}The nonlinear energy functional \eqref{energy}, defined on $D(E)=H^1(\GG)$, is the conserved energy associated to the NLS  equation compactly written as 
\begin{equation}\label{nls}i \frac{d}{dt} \Psi(t) = H\Psi(t) - |\Psi(t)|^{2\mu}\Psi(t). \end{equation}
Where  $H$ is the linear self-adjoint operator on $\GG$ associated to the quadratic form $E^{lin}$ with domain $D(E^{lin})=H^1(\GG)$, we refer to Sec. \ref{s:prel}, Eqs. \eqref{domain} - \eqref{action} for its  rigorous definition. Here we just remark that  $H$  encodes both the presence of the potential $W$ and the presence of delta-interactions of strength $\alpha_v$ in the vertices.  By the definition of $E_0$, one has $-E_0 = \inf\sigma(H)$, $\sigma(H)$ being the spectrum of $H$. \\ 
The nonlinear term $|\Psi|^{2\mu}\Psi $ in Eq. \eqref{nls} must be understood componentwise as $(|\Psi|^{2\mu}\Psi)_e = |\psi_e|^{2\mu}\psi_e$ for every edge $e$,  where $\psi_e$ is the component of the wavefunction $\Psi$ on the edge $e$ (see Sec. \ref{s:prel} for the details).  Hence, Eq. \eqref{nls} is understood as a single-particle equation on a one-dimensional ramified structure. On each branch (edge), the dynamics   is  governed by: a dispersive term (the second order spatial derivative in Eq. \eqref{action}); plus  a potential term; plus   delta-type interactions in the vertices (which can be understood as singular potentials);  plus  a (focusing) nonlinear power-type term.  

 We recall that, for any initial datum $\Psi(0)= \Psi_0\in H^1(\GG)$, the energy $E[\Psi(t)]$ and the mass $\|\Psi(t)\|^2$ are conserved along the flow associated to Eq. \eqref{nls}. Moreover: if $0<\mu<2$ then  Eq. \eqref{nls} (in weak form) is globally well-posed in $H^1(\GG)$;  if $\mu=2$ global well-posedness holds true for small enough mass ($\|\Psi_0\|^2<\sqrt3/K_{6,2}^3(\GG)$). We refer to \cite{cacciapuoti-finco-noja-rxv16} for the precise statements and the proofs. 

\subsection*{Stationary states and bifurcations} It is well known that, whenever a ground state exists, it is a stationary solution of Eq. \eqref{nls}, i.e., a solution of the form $\Psi(t) = e^{i\omega t} \Psi(\omega)$, with $\omega\in\RE$. The function $\Psi(\omega)$ satisfies the  stationary equation
\begin{equation}\label{stationary}H\Psi(\omega) - |\Psi(\omega)|^{2\mu}\Psi(\omega)  = -  \omega \Psi(\omega), \end{equation}
where the parameter $\omega\in\RE$ must be chosen in order to satisfy the mass constraint. Solutions of Eq. \eqref{stationary} are called \emph{stationary states}. 

In the subcritical case, for mass small enough (possibly smaller than $m^*_\mu$ in Eq. \eqref{m_mustar}), one can prove that the ground state $\hat\Psi$ is the solution of Eq. \eqref{stationary} bifurcating from the null state along the direction of the eigenvector of $H$ corresponding to the eigenvalue $-E_0$.  The bifurcation occurs for $\omega = E_0$. We refer to \cite{cacciapuoti-finco-noja-rxv16} for the details.

Eq. \eqref{stationary} has an interest in its own as its solutions identify the critical points of the energy functional  \eqref{energy} subject to the mass constraint. 

\subsection*{Star-graph with delta-interaction in the vertex} A first rigorous analysis of Eq. \eqref{stationary} was performed for the case of a star-graph with $N$ edges, when $W=0$, and $\alpha \in\RE$,  in \cite{adami-cacciapuoti-finco-noja-epl12} (see also \cite{adami-cacciapuoti-finco-noja-jde14}). In such a case all the stationary states can be explicitly computed.

In the same setting, the minimization problem \eqref{minimization}, for $0<\mu\leq 2$, was studied in \cite{adami-cacciapuoti-finco-noja-aihp14}. One main result in \cite{adami-cacciapuoti-finco-noja-aihp14} is that, for $\alpha<0$ (attractive interaction in the vertex), by exploiting the explicit form of the stationary states, it is possible to identify the ground state with the unique (up to phase multiplication) symmetric stationary state. Several techniques and ideas, that can also be applied to generic graphs, such as the interplay between the concentration-compactness lemma and condition \eqref{sufficient}, were first used  in \cite{adami-cacciapuoti-finco-noja-aihp14}. We remark that in Th. \ref{t:main} the value of the threshold mass  $m^*_\mu$ is slightly improved with respect to the one  given in \cite{adami-cacciapuoti-finco-noja-aihp14}. 

In general it is not possible to say what happens when $m$ crosses the mass threshold $m^*_\mu$. For a star-graph with $0<\mu<2$, $W=0$, and $\alpha<0$, it was proved in \cite{adami-cacciapuoti-finco-noja-jde16} that for mass large enough the ground state does not exist.
Despite that, one can show that when the mass is larger than a certain threshold, the symmetric stationary state is a local minimum of the energy functional \eqref{energy} constrained to the manifold of fixed mass, see \cite{adami-cacciapuoti-finco-noja-jde16}.

\subsection*{Free Laplacian with Kirchhoff conditions in the vertices} The case $W = 0$ and  $\alpha_v = 0$ for all $v\in \Ver$ requires a separate discussion, as in this case  Ass.  \ref{a:E0} is not satisfied. We shall use the following notation $E_{(0,0)}[\Psi] = E_{W=0,\alpha_v=0}[ \Psi] $,  and  $E_{(0,0)}^{lin}[\Psi] = E_{W=0,\alpha_v=0}^{lin} [ \Psi]$.  

We recall that the Hamiltonian $H_{(0,0)}$ associated to the quadratic form $E_{(0,0)}^{lin}$ is still defined as in Eqs. \eqref{domain} and \eqref{action}, with $W=0$ and $\alpha_v=0$. When $\alpha_v=0$, the gluing conditions encoded in the definition of $D(H)$ are usually referred to as \emph{Kirchhoff} (or \emph{standard}) conditions.

As a first remark we note that 
\[\inf \left\{ E_{(0,0)}^{lin} [ \Psi] \,|\;  \Psi\in H^1(\GG),\; \|\Psi\|^2 = 1\right\}=0,\]
hence, Ass.  \ref{a:E0} is not satisfied and Th. \ref{t:main} does not give any information on the existence of the ground state (Th. \ref{t:lowerbound} still hods true though). 

A first result on the minimization problem \eqref{minimization} for $E_{(0,0)}[\Psi] $ was given in \cite{adami-cacciapuoti-finco-noja-jpa12}, where it was shown that on a star-graph and  in the cubic case ($\mu=1$), the minimizer does not exist for any value of the mass. The stationary states were explicitly computed in \cite{adami-cacciapuoti-finco-noja-epl12}, see also \cite{adami-cacciapuoti-finco-noja-jde14}. 

A systematic analysis  of  the minimization problem for  $E_{(0,0)}$  on generic graphs has been performed in \cite{adami-serra-tilli-mtn15,adami-serra-tilli-cv15,adami-serra-tilli-jfa16} for the subcritical case (see  \cite{adami-mmnp16} for the analysis of the cubic case), and in \cite{adami-serra-tilli-cmp16} for the critical case. One main result in the series of works \cite{adami-serra-tilli-mtn15,adami-serra-tilli-cv15,adami-serra-tilli-jfa16} is the identification of a topological condition (called Assumption H) that excludes (apart for very specific examples of graphs) the existence of the ground state for any value of the mass $m$. Assumption H can be stated as: the graph $\GG$ can be covered by cycles (here the $\infty$-points of the external edges, see. Figs. \ref{f:1} and \ref{f:2}, are regarded as a single vertex). One example of graph to which Assumption H applies is the star-graph. 

It is worth noticing that this result is very unstable under perturbations of the energy functional, in the sense that adding any arbitrarily small negative potential may  turn the infimum in \eqref{E0} into  strictly negative. Hence, as a consequence of Th. \ref{t:main}, the ground state would exist for small mass, despite the topological condition. This is exactly what happens for a star-graph with $N$ edges: for arbitrary $\alpha<0$ one has $E_0 = |\alpha|^2/N^2$ see \cite{adami-cacciapuoti-finco-noja-aihp14}; on the other hand, for $\alpha =0$ and $N\geq 3$ no ground exists for any $m>0$. 
%

Concerning the critical case, we remark that in \cite{adami-serra-tilli-cmp16}, among other results, the authors prove that for a large class of graphs (e.g., the ones that  do not satisfy Assumption H, have no terminal edge, and have at least two external edges) the ground state exists if and only if $m\in[\sqrt3/K_{6,2}^3(\GG),\sqrt3/K_{6,2}^3(\RE)]$. In view of  Ths. \ref{t:lowerbound} and \ref{t:main}, the existence of the lower bound for the mass parameter might be surprising, so we briefly comment on it. 

One issue in the minimization of $E_{(0,0)}$ for small mass, is that the infimum might be zero, and never attained because  minimizing sequences are \emph{vanishing} in the sense of Lemma \ref{l:cc}. This is indeed the case when $\GG=\RE$, the mass threshold being $\pi\sqrt3/2$, see Eq. \eqref{nuRE}. In the presence of potential terms or delta-interactions, if Ass. \ref{a:E0} is satisfied then  $-\nu_\mu(m)$ is  strictly negative by Th. \ref{t:main}; hence, vanishing cannot occur and there is no lower bound on $m$. 

On the other hand, for large mass, exactly larger than $\sqrt3/K_{6,2}^3(\GG)$, one has a different issue: the Gagliardo-Nirenberg inequality does not guarantee that the infimum in \eqref{minimization} is lower bounded. In \cite{adami-serra-tilli-cmp16}, the authors show that by the topological assumptions,  $-\nu_\mu(m)$ is indeed   lower bounded (and strictly negative due to the large mass)  thus implying that the ground state exists. 


\subsection*{Other related works} For several specific examples of graphs Eq. \eqref{stationary} can be explicitly solved. One interesting case is the tadpole-graph, see Fig. \ref{f:2}, for $E = E_{(0,0)}$. The stationary states for the tadpole-graph  have been completely characterized in \cite{cacciapuoti-finco-noja-pre15}. Certain families of solutions can be understood in terms of bifurcation theory from embedded eigenvalues and threshold resonances. Bifurcations and stability properties have been further analyzed in \cite{noja-pelinovsky-shaikhova-nl15}. The existence of the ground state for the tadpole-graph for any $m>0$ has been proved in \cite{adami-serra-tilli-cv15}, see also \cite{adami-serra-tilli-jfa16}. A general approach to the study of the stationary solutions has been recently proposed in \cite{gnutzmann-waltner-pre16_1,gnutzmann-waltner-pre16_2}. The stationary solutions on a compact star-graph, in a setting in which the nonlinear term changes from edge to edge has been studied in \cite{sabirov-sobirov-babajanov-matrasulov-pla13,sobirov-sabirov-matrasulov-rxv11}. 

A similar analysis for stationary states on periodic graphs is in \cite{pelinovsky-schneider-ahp17} (see also \cite{gilg-pelinovsky-schneider-ndea16}). While the ground state for the dumbbell-graph is studied in  \cite{marzuola-pelinovsky-amre16}. 

Existence/nonexistence of the ground state in a slightly  different setting, i.e., when the nonlinearity is supported only on a compact region of the graph, has been investigated in \cite{serra-tentarelli-jde16,serra-tentarelli-na16,tentarelli-jmaa16}. The same model was first proposed in   \cite{gnutzmann-smilansky-derevyanko-pra11}, to study the scattering through  a nonlinear network. \\ 

The paper is structured as follows. In Section \ref{s:prel} we set up the model and recall several preliminary results, included the Concentration-Compactness Lemma \ref{l:cc}. Section \ref{s:3} is devoted to the proof of Th. \ref{t:lowerbound}. Section \ref{s:4} is devoted to the proof of Th. \ref{t:main}. We conclude the paper with a short appendix  (App. \ref{s:appendix}) in which we prove that in the subcritical case,  if the potential $W$ decays at infinity, 
$-t_\mu(m)$ is an upper bound for  the infimum in \eqref{minimization}.
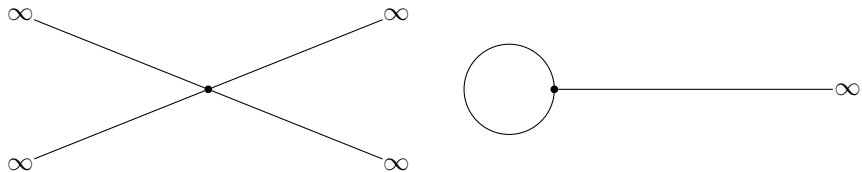
\begin{figure}[t]
\begin{tikzpicture}
\node at (-10.5,0) [infinito](1) {${\infty}$};
\node at (-10.5,2) [infinito](2) {${\infty}$};
\node at (-5.5,2) [infinito](3) {${\infty}$};
\node at (-5.5,0) [infinito](4) {${\infty}$};
\node at (-8,1) [nodo](0) {};
\draw [-,black] (0) -- (1);
\draw [-,black] (0) -- (2);
\draw [-,black] (0) -- (3);
\draw [-,black] (0) -- (4);
\node at (-3.4,1) [nodo](5) {};
\node at (0.5,1) [infinito](6) {${\infty}$};
\draw [-,black] (-4,1) circle (0.6cm) ;
\draw [-,black] (5) -- (6);
\end{tikzpicture}
\caption{\label{f:2}On the right, a star-graph with  $4$ edges. On the left a tadpole graph.}
\end{figure}

\section{Preliminaries}
\label{s:prel}
In this section we recall several basic definitions and facts about  metric graphs, moreover we recall the Gagliardo-Nirenberg inequalities and the concentration-compactness lemma. 

We denote by $\Ver$ the set of vertices of the graph $\GG$, and by $\Ver_-$ the set of vertices for which the coupling constant $\alpha_v$ is strictly negative.  
We denote by $\Ed$ the set of edges of the graph and decompose it as $\Ed =\Ed^{in}\cup  \Ed^{ext} $. The set $\Ed^{ext}$ denotes the set of external edges, these are the edges that start from one vertex of the graph and extend to infinity; which means that every element of $\Ed^{ext}$ can be identified with the half-line $ [0,+\infty)$. The set  $\Ed^{in}$ denotes the set of internal edges of the graph, these are the edges of finite length; each edge  $e\in \Ed^{in}$ can be identified with a segment $[0,\ell_e]$, with   $\ell_e$ denoting its  length.

In what follows we shall use the notation $I_e \equiv [0,+\infty)$ for $e\in \Ed^{ext}$, and $I_e \equiv [0,\ell_e]$ for $e\in \Ed^{in}$.

A (complex valued) function on $\GG$ is a map $\Psi : \GG \to \CO^{|\Ed|}$, to be understood as $\Psi =  \bigoplus_{e\in \Ed}\psi_e$ with $\psi_e :I_e \to \CO$, $\psi_e$ denoting the wave function component on the edge $e$. 

For $p\in[1,+\infty]$, one has  $L^p(\GG) = \bigoplus_{e\in \Ed}L^p(I_e)$,  and we denote by $\|\cdot \|_p$ the corresponding norm 
\[\|\Psi\|_p^p = \sum_{e\in \Ed} \|\psi_e\|_{L^p(I_e)}^p\,, \quad  p\in[1,+\infty)\,;\qquad  \|\Psi\|_\infty = \max_{e\in \Ed} \|\psi_e\|_{L^\infty(I_e)}.\]
For $p = 2$ we shall denote the norm in $L^2(\GG)$ simply by $\|\cdot\|$.

We also recall the definition of  the Sobolev spaces $H^1(\GG)$ and $H^2(\GG)$.  Denote by $C(\GG)$ the set of continuous functions on $\GG$, then 
\begin{equation}\label{H1}
H^1(\GG)  = \left\{\Psi\in C(\GG) \,|\;  \psi_e \in H^1(I_e) \; \forall e \in \Ed\right\},
\end{equation}
equipped with the norm 
\[\| \Psi \|_{H^1(\GG)}^2  = \sum_{e\in \Ed} \| \psi_e \|_{H^1(I_e)}^2;
\]
and 
\[H^2(\GG)  = \left\{\Psi\in H^1(\GG) \,|\; \psi_e \in H^2(I_e) \; \forall e \in \Ed\right\}\]
equipped with the norm 
\[
\| \Psi \|_{H^2(\GG)}^2  =  \sum_{e\in \Ed} \| \psi_e \|_{H^2(I_e)}^2.
\]

We recall the Gagliardo-Nirenberg inequalities on graphs.
\begin{proposition}[Gagliardo-Nirenberg inequalities on graphs]\label{p:GN} For any connected graph $\GG$ such that $|\Ed|<\infty$ and $|\Ver| <\infty$, and for any    $p,q\in[2, +\infty]$, with $p\geq q$, and $\alpha = \frac{2}{2+q}(1-q/p)$, there exist two constants $K_{p,q}(\GG)$ and $\KK_{p,q}(\GG)$ such that 
\begin{equation}\label{gn1}
\| \Psi \|_{p} \leq K_{p,q}(\GG) \| \Psi'\|^\alpha \| \Psi \|^{1-\alpha}_q \qquad \text{if }  |\Ed^{ext}|\geq 1,
\end{equation}
\begin{equation}\label{gn2}
\| \Psi \|_{p} \leq \KK_{p,q}(\GG) \| \Psi\|^\alpha_{H^1} \| \Psi \|^{1-\alpha}_q  \qquad \text{if }  |\Ed^{ext}|= 0,
\end{equation}
for all $\Psi \in H^1(\GG)$. 
\end{proposition}

A proof of inequality \eqref{gn1} is in  \cite{adami-serra-tilli-jfa16} (see also \cite{adami-cacciapuoti-finco-noja-aihp14, adami-cacciapuoti-finco-noja-jde14, haeseler-rxv11,tentarelli-jmaa16}). In the case of compact graphs ($|\Ed^{ext}|= 0$) inequality  \eqref{gn1} cannot hold true (it is clearly violated by the constant function). Nevertheless, it can be replaced by the weaker inequality \eqref{gn2}, for  a proof we refer to   \cite{mugnolo14}. 

We recall some facts about the quadratic form $E^{lin}[\Psi]$ defined in Eq. \eqref{Elin}. We always consider $E^{lin}$ on the domain $D(E^{lin}) = H^1(\GG)$. We note that on such  domain,  $\Psi(v)$ (the value of the wave  function in a vertex of the graph) is well defined due to the global  continuity condition, see Eq.  \eqref{H1}.

Moreover  by using Gagliardo-Nirenberg inequalities it is easy to prove (see \cite{cacciapuoti-finco-noja-rxv16}) that for $W\in L^1(\GG) +L^{\infty}(\GG)$ 
\begin{equation*}
\Big|  (\Psi,W\Psi) +    \sum_{v\in \Ver} \alpha(v) |\Psi(v)|^2\Big| \leq a \|\Psi'\|^2 + b \|\Psi\|^2, \qquad \text{with }0<a<1, \, b>0,
\end{equation*}
which, by KLMN theorem, implies that the form $E^{lin}$ is lower bounded ($E_0<+\infty$) and  closed,  hence defines a selfadjoint operator. This a standard result for Schr\"odinger operators on the real-line,  see, e.g., \cite[Ch. 11.3]{lieb-loss01}.
 
It is easy to prove that the self-adjoint operator corresponding to the quadratic form $E^{lin}$  coincides with the Hamiltonian $H:D(H)\subset L^2(\GG) \to L^2(\GG)$ defined on the domain
\begin{equation}\label{domain}
D (H) :=  \left\{\Psi \in H^2(\GG) \,\Big|\;   \sum_{ e\prec v} \partial_o \psi_e (v)= \alpha_v \Psi (v)\quad \forall v \in \Ver
\right\},
\end{equation}
where $ e\prec v$ denotes the set of edges having at least one endpoint identified with the vertex $v$, and  we have denoted by $\partial_o$ the outward derivative from the vertex (more precisely $\partial_o \psi_e (v) = \psi_e' (0)$ when  $\psi_e(v)$ is identified by $\psi_e(0)$; $\partial_o \psi_e (v) = -\psi_e' (\ell_e)$ when  $\psi_e(v)$ is identified by $\psi_e(\ell_e)$). The action of $H$ is defined by  
\begin{equation}\label{action}
 (H\Psi)_e = - \psi_e '' + W_e \psi_e \qquad \forall e\in \Ed. 
\end{equation}

We conclude this section by recalling the  concentration-compactness lemma  (see Lem. \ref{l:cc} below) that will be needed in the proof of Th. \ref{t:main}. For the proof we refer to \cite{cacciapuoti-finco-noja-rxv16} (see also \cite{adami-cacciapuoti-finco-noja-aihp14}). 

For any $y\in\GG$ and $t>0$, we denote by $B(y,t)\subset \GG$ the open ball of radius $t$ and center $y$
\[
B(y,t):=\{x\in\GG\,|\; d(x,y)<t\},
\]
here 
$d(x,y)$ denotes the distance between two points of the graph, defined as the infimum of the length of the paths connecting $x$ to $y$. 

For any function $\Psi\in L^2(\GG)$ and $t> 0$ we define the concentration function
$\rho(\Psi,t)$ as
\[
\rho(\Psi,t) : = \sup_{y\in\GG} \|\Psi\|_{L^2(B(y,t))}^2;
\]
and, for any sequence $\{\Psi_n\}_{n\in\NA}$, $\Psi_n \in L^2(\GG)$, the concentrated mass
parameter $\tau$ as  
\begin{equation*}
\tau := \lim_{t\to+\infty} \liminf_{n\to\infty} \rho(\Psi_n,t).
\end{equation*}
The parameter  $\tau$ plays a key role in the
concentration-compactness lemma because it distinguishes the
occurrence of vanishing, dichotomy or compactness in $H^1(\GG)$-bounded
sequences in a sense precisely defined below.
\begin{lemma}[Concentration-compactness]
\label{l:cc}
Let $m>0$ and $\{\Psi_n\}_{n\in\NA}$ be such
that: $\Psi_n\in H^1(\GG)$,
\begin{equation*}
\|\Psi_n\|^2 \to m \quad \text{as} \quad n\to \infty\,,
\end{equation*}
\begin{equation*}
\sup_{n\in\NA}\|\Psi_n'\|<\infty\,.
\end{equation*}
Then there exists a subsequence $\{\Psi_{n_k}\}_{k\in \NA}$ such that:
\begin{enumerate}[i)]
\item\label{i:cc1}
(Compactness) If $\tau=m$, at least one of
    the two following cases occurs:
\begin{itemize}
\item[$i_1)$](Convergence) There exists a
function
$\Psi\in H^1(\GG)$  such that $\Psi_{n_k}\to \Psi$  in $L^p$ as
$k\to\infty$ for all $2\leq p\leq \infty$ .
\item[$i_2)$]\label{i:runaway}(Runaway) There exists $e^*\in \Ed^{ext}$, such that for all $t>0$,  and
$2\leq p\leq\infty$
\begin{equation}
\label{e:run-1}
\lim_{k\to \infty}  
\left(\sum_{e\neq e^*}\|(\Psi_{n_k})_e\|_{L^p(I_e)}^p   + \|(\Psi_{n_k})_{e^*}\|_{L^p((0,t))}^p \right) =0 .
\end{equation}
\end{itemize}
\item\label{i:cc2} (Vanishing) If $\tau=0$, then  $\Psi_{n_k}\to 0 $ in $L^p$ as
$k\to\infty$  for all $2< p\leq \infty$. 
\item\label{i:cc3} (Dichotomy) If $0<\tau<m$, then there exist two sequences 
$\{\VV_k\}_{k\in\NA}$ and $\{\WW_k\}_{k\in\NA}$ in $H^1(\GG)$
such that 
\begin{equation}
 \supp \VV_k \cap \supp \WW_k = \emptyset 
\label{dic1} 
\end{equation}
\begin{equation}
|\VV_k(\x)| + |\WW_k(\x)| \leq |\Psi_{n_k}(\x)| \qquad \forall \x\in\GG 
\label{dic2} 
\end{equation}
\begin{equation}
\| \VV_k \|_{H^1(\GG)} + \|\WW_k\|_{H^1(\GG)} \leq c
\|\Psi_{n_k}\|_{H^1(\GG)}
\label{dic3} 
\end{equation}
\begin{equation}
\lim_{k \to \infty} \|\VV_k \|^2 = \tau \qquad \qquad \lim_{k \to \infty} \|
\WW_k \|^2= m -\tau
\label{dic4} 
\end{equation}
\begin{equation}
\liminf_{k\to \infty} \left( \|\Psi_{n_k}'\|^2 - \| \VV_k' \|^2 - \|
\WW_k' \|^2 \right) \geq 0
\label{dic5}
\end{equation}
\begin{equation}
\lim_{k\to \infty} \left( \|\Psi_{n_k} \|_{p}^p - \| \VV_k \|_{p}^p - \|
\WW_k \|_{p}^p \right) =0 \qquad 2 \leq p < \infty
\label{dic6}
\end{equation}
\begin{equation}
\label{dic7}
\lim_{k\to\infty}\left\||\Psi_{n_k}|^2-  |\VV_{k}|^2 -| \WW_{k}|^2\right\|_\infty=0.
\end{equation}
\end{enumerate}
\end{lemma}

\section{\label{s:3}Proof of Theorem \ref{t:lowerbound}}
We always consider the functional  $E$ on the domain $D(E) = H^1(\GG)$ and note that 
\begin{equation*}
E[\Psi]=E^{lin}[\Psi ] -\frac{1}{\mu+1} \|\Psi \|_{2\mu+2}^{2\mu + 2} .
\end{equation*}
\begin{proof}[Proof of Th. \ref{t:lowerbound}]
We note the trivial lower bound 
\begin{equation}\label{trivialbound}
E[\Psi] \geq 
 \| \Psi ' \|^2 - (\Psi,W_-\Psi) -   \sum_{v\in \Ver_-} |\alpha_v| |\Psi(v)|^2 -\frac{1}{\mu+1} \|\Psi \|_{2\mu+2}^{2\mu + 2}.  
 \end{equation} 
Write $W_- = W_{-,1}+W_{-,\infty}$.  By Gagliardo-Nirenberg inequality and setting $\|\Psi\|^2 =m$, one has the following lower bound for $E[\Psi]$,
\begin{equation}\label{lower}
 E[\Psi] 
 \geq  \| \Psi' \|^2  - a_\mu m^{\frac{2+\mu}{2}}\|\Psi'\|^{\mu}
  - b \sqrt m  \|\Psi'\| - c m ,
\end{equation}
with 
\begin{equation*}
a_\mu =  \frac{K^{2\mu+2}_{2\mu+2,2}}{\mu+1}  ; \quad b = K_{\infty,2} \left( \sum_{v\in \Ver_-}|\alpha_v| +\|W_{-,1}\|_1 \right)   ;\quad c =  \|W_{-,\infty}\|_{\infty} .
\end{equation*}
Next we distinguish two cases. For any $0<\mu<2$ and $m>0$, we note that there exists $\beta \equiv \beta(a_\mu,b,c,m)>0$,    such that $x^2 - a_\mu m^{\frac{2+\mu}{2}} x^\mu  -b  \sqrt m  x - c m   \geq - \beta $. Hence, 
 \[ 
 E[\Psi] \geq  - \beta   \qquad 0<\mu<2,\; m>0 ,
\] 
so that it must be  $\nu_\mu(m) \leq \beta $.\\
For $\mu = 2$ and $0<m<m_2^*= 1/\sqrt{a_2}$,  we note that $(1-a_2m^2)x^2 -b  \sqrt m  x - c m   \geq - b^2m/(4(1-a_2m^2)) - cm$,  hence 
 \[ 
 E[\Psi] \geq    - b^2m/(4(1-a_2m^2)) - cm   \qquad \mu = 2,\; 0<m<m_2^*,
\]
so that $\nu_\mu(m) \leq b^2m/(4(1-a_2m^2)) + cm $ under the same conditions on $\mu$ and $m$.
\end{proof}


\section{\label{s:4}Proof of Theorem \ref{t:main}}
In this section we prove the existence of the ground state for small enough mass. Since the problem in the subcritical case $0<\mu<2$ was discussed in \cite{cacciapuoti-finco-noja-rxv16}, in the proof we often  skip the details whenever the argument used in \cite{cacciapuoti-finco-noja-rxv16} remains unchanged. 
\begin{proof}[Proof of Th. \ref{t:main}]
We start by recalling that  the bound 
\begin{equation}\label{nu-cond1}
-\nu_\mu(m)\leq -mE_0
\end{equation} 
follows directly from the inequality  $E[\Psi] < E^{lin}[\Psi]$, for all $\Psi\in H^1(\GG)$. Hence, by Ass. \ref{a:E0}, it must be $\nu_\mu(m)>0$. 

\medskip

In the  remaining part of the proof we shall show  that, both in the subcritical and in the critical case,   for $m<m_\mu^*$ minimizing sequences have a  convergent subsequence (case $i_2$ of Lemma \ref{l:cc}). 

We recall that, by Th. \ref{t:lowerbound}, we already know that $\nu_\mu(m)<+\infty$ for any $m>0$ if $0<\mu<2$, or for $m<m_2^*$  if $\mu=2$, hence the existence of a lower bound for the infimum in Eq. \eqref{minimization} is granted for $m<m_\mu^*$ for any $0<\mu\leq2$.

Let $\{\Psi_n\}_{n\in\NA}$ be a minimizing sequence, i.e., $\Psi_n \in H^1(\GG)$, $\|\Psi_n\|^2=m$, and $\lim_{n\to\infty} E[\Psi_n] = -\nu_\mu(m)$. We remark that, when choosing a minimizing sequence, it is enough to assume $\|\Psi_n\|^2\equiv m_n\to m$ as $n\to \infty$, since in such a case one can define $\widetilde \Psi_n = \sqrt{m} \Psi_n /\|\Psi_n\|$ and use the fact that $\lim_{n\to\infty} E[\widetilde\Psi_n] = \lim_{n\to\infty} E[\Psi_n] $. 

We shall prove that, for any $0<\mu\leq2$ and   $m$ small enough,  there exists $\hat \Psi \in H^1 (\GG)$ such that $\|\hat \Psi\|^2 = m$, $E[\hat \Psi] =-\nu_\mu(m)$ and $\Psi_n \to \hat \Psi$ in $ H^1 (\GG)$.

We claim that, up to taking a subsequence that we still denote by $\Psi_n$, the following bound holds true
\begin{equation}\label{world}
\sup_{n\in\NA}\|\Psi_n'\|<\infty,
\end{equation}
for any $m>0$ if $0<\mu<2$, and for $0<m<m_2^*$ if $\mu=2$. 

To prove the bound \eqref{world} we start by noticing that, up to taking a subsequence, we can  assume that 
\begin{equation}\label{flesh}0<m_n<\sqrt{1+\eta}\,m\quad\text{and}\quad E[\Psi_n] \leq -\nu_\mu(m)/2,\end{equation}
for all $\eta>0$. 
Next we consider first the critical case $\mu=2$. Fix $0<m<m_2^*= 1/\sqrt{a_2}$. Then there exists $0<\eta<1/2$ such that $m<\sqrt{\frac{1-\eta}{(1+\eta)a_2}}$. In inequality \eqref{lower} we set $\Psi \equiv \Psi_n$ and note that, by  \eqref{flesh}, 
 \[ 
 E[\Psi_n] \geq\eta \|\Psi_n'\|^2 +(1-\eta -(1+\eta)a_2m^2  )    \|\Psi_n'\|^2 - b (1+\eta)^{1/4} \sqrt m  \|\Psi_n'\| - c \sqrt{1+\eta}\,m.
\]
By the trivial bound $Ax^2-Bx-C\geq -\frac{B^2}{4A} - C$, for all $A,B,C>0$, we infer 
\begin{equation*}
 E[\Psi_n] \geq\eta \|\Psi_n'\|^2 - \frac{ b^2 \sqrt{1+\eta} \, m}{4(1-\eta -(1+\eta)a_2m^2)} - c \sqrt{1+\eta}\,m,
\end{equation*}
for any $m<\sqrt{\frac{1-\eta}{(1+\eta)a_2}}$. The latter bound, together with the fact that $E[\Psi_n]<0$ by \eqref{flesh}, implies the claim \eqref{world}. \\
For $0<\mu<2$ we proceed in a similar way.  By inequality \eqref{lower}, we infer that for all $m>0$ there exists $\tilde\beta \equiv \tilde\beta(a_\mu,b,c,m)>0$,    such that $(1-\eta)x^2 - a_\mu m_n^{\frac{2+\mu}{2}} x^\mu  -b  \sqrt m_n  x - c m_n   \geq - \tilde\beta $. Hence, 
 \[ 
 E[\Psi_n] \geq \eta \|\Psi_n'\|^2   - \tilde \beta   \qquad 0<\mu<2,\; m>0 ,\]
 from which the claim \eqref{world}  follows. 
 
The bound \eqref{world}, together with Gagliardo-Nirenberg inequality and the first bound in \eqref{flesh}, implies 
\[
\sup_{n\in\NA}\|\Psi_n\|_p <\infty \qquad \forall p\in[2,+\infty].\]
Moreover  the following lower bound holds true
\begin{equation}
\label{e:floor}
 \frac{1}{\mu+1} \| \Psi_n\|_{2\mu+2}^{2\mu+2} +(\Psi_n,W_-\Psi_n) +  \sum_{v \in \Ver_- } |\alpha_v| |\Psi_{n}(v)|^2 \geq
\frac{\nu_\mu(m)}2 \,.
\end{equation}
The latter  is an immediate consequence of the bounds \eqref{trivialbound} and \eqref{flesh}.

Next we use Lem. \ref{l:cc} and  prove that  vanishing and dichotomy
cannot occur for $\{\Psi_n\}_{n\in\NA}$. Set $\tau = \lim_{t\to\infty}\liminf_{n\to\infty}
\rho(\Psi_n,t)$. 
 
If
$\tau=0$,  then by Lem. \ref{l:cc} the l.h.s. in Eq. \eqref{e:floor} would converge to zero bringing to a contradiction (see \cite{cacciapuoti-finco-noja-rxv16} for the details), hence $\tau>0$.

Suppose $0<\tau<m$, then there
would exist $\VV_k$ and $\WW_k$ satisfying \eqref{dic1}-\eqref{dic7}. It is possible to prove (see \cite{cacciapuoti-finco-noja-rxv16} for the details) that in this  case it must be 
\[
\liminf_{k\to\infty} \left(
E[\Psi_{n_k} ] - E[ \VV_k] - E[\WW_k]
\right) \geq 0 \,,
\]
which implies
\begin{equation}
\label{e:black-1}
\limsup_{k\to\infty} \left(
 E[ \VV_k] + E[\WW_k]
\right) \leq -\nu_\mu(m) \,.
\end{equation}
We use the  identity
\[
E[\Psi] = \frac{1}{\delta^2} E[\delta \Psi] + \frac{\delta^{2\mu} -1}{\mu+1} \| \Psi
\|_{2\mu+2}^{2\mu+2},
\]
which holds true for any  $\Psi\in H^1(\GG)$ and $\delta >0$. Let $\delta_k= \sqrt {m} / \|\VV_k\|$ and $\gamma_k = \sqrt {m} / \|\WW_k\|$  so that $\|\delta_k \VV_k\|^2,\,
\|\gamma_k \WW_k\|^2 =m$. Then, using the above identity and the fact that
$E[\delta_k \VV_k], E[\gamma_k \WW_k] \geq -\nu_\mu(m)$,  one has
\[
E[\VV_k] \geq - \frac{\nu_\mu(m)}{\delta^2_k} + \frac{\delta^{2\mu}_k -1}{\mu+1} \| \VV_k
\|_{2\mu+2}^{2\mu+2}
\]
\[
E[\WW_k] \geq - \frac{\nu_\mu(m)}{\gamma^2_k} + \frac{\gamma^{2\mu}_k -1}{\mu+1} \| \WW_k
\|_{2\mu+2}^{2\mu+2}
\]
from which 
\[
E[\VV_k]+E[\WW_k] \geq -\nu_\mu(m) \left( \frac{1}{\delta^2_k} + \frac{1}{\gamma^2_k} \right) +
\frac{\delta^{2\mu}_k -1}{\mu+1} \| \VV_k \|_{2\mu+2}^{2\mu+2} +
\frac{\gamma^{2\mu}_k -1}{\mu+1} \| \WW_k \|_{2\mu+2}^{2\mu+2}\,.
\]
Notice that,   by \eqref{dic4},   $\delta^2_k \to m/\tau $ and  $\gamma^2_k \to 1/(1-\tau/m)$, hence  $1/ \delta^2_k+1/\gamma^2_k \to 1$. Moreover set  $\theta = \min \{ (\tau/m)^{-\mu} , (1-\tau/m)^{-\mu} \}>1$. Then  
\begin{align}
\label{e:black-2}
\liminf_{k\to\infty} \left(
 E[ \VV_k] + E[\WW_k]
\right) 
&\geq -\nu_\mu(m) + \frac{\theta -1}{\mu+1} \liminf_{k\to\infty} \| \Psi_{n_k}
\|_{2\mu+2}^{2\mu+2} > -\nu_\mu(m),
\end{align}
where we used the fact that $\liminf_{k\to\infty} \| \Psi_{n_k}
\|_{2\mu+2}^{2\mu+2} \neq 0$. The latter claim is proved by noticing
that $\liminf_{k\to\infty} \| \Psi_{n_k} 
\|_{2\mu+2}^{2\mu+2} = 0$ would bring to a contradiction with inequality \eqref{e:floor}. This can be understood by using the inequalities 
\[\|\Psi\|_\infty \leq K_{\infty,2\mu+2} \|\Psi'\|^{\alpha}\|\Psi\|_{2\mu+2}^{1-\alpha}\]
with $\alpha = 1/(\mu+2)$; and 
\begin{equation}\label{W-1}
(\Psi,W_- \Psi)  \leq \|W_-\|_r \|\Psi\|^2_{2r/(r-1)} \leq K^2_{\frac{2r}{r-1},2\mu+2} \|W_-\|_r  \|\Psi'\|^{2\alpha}\|\Psi\|_{2\mu+2}^{2(1-\alpha)}
\end{equation}
which holds true for all $r\in[1,1+1/\mu]$ and with $ \alpha= (1-(\mu+1)(r-1)/r)/(\mu+2)$. 

Since inequalities  \eqref{e:black-1} and \eqref{e:black-2} cannot be satisfied at the same time we must also exclude the case $0<\tau<m$. Hence it must be $\tau=m$.
%
%

Next we prove that  for $m$ small enough  the minimizing sequence is not {\em runaway}. 
By absurd suppose that $\{\Psi_n\}_{n\in\NA} $ is {\em runaway}, then we have that 
\begin{equation}\label{limit}
\lim_{n\to\infty} \Psi_{n} (v) =0\quad \forall v\in \Ver\qquad \text{and}\qquad \lim_{n\to\infty}(\Psi_n, W_- \Psi_n)= 0.
\end{equation} The first limit  is a direct consequence of Lem. \ref{l:cc}, Eq. \eqref{e:run-1}. The proof of the second one requires a bit more work.
Assume that $\Psi_n$ escapes at infinity on the external  edge $e^*$ (this can always be done up to taking a subsequence). We note that
\begin{equation*}
\lim_{n\to\infty}\int_{I_e} (W_-)_e |(\Psi_n)_e|^2 dx = 0 \qquad \forall e \neq e^* ,
\end{equation*}
this is a  direct consequence of the inequality 
\[ 
\int_{I_e} (W_-)_e |(\Psi_n)_e|^2 dx  \leq \|W_{-,\infty}\|_{\infty} \|(\Psi_n)_e\|^2_{L^2(I_e)} + \|W_{-,1}\|_{1} \|(\Psi_n)_e\|^2_{L^\infty(I_e)} \]
and  Lemma \ref{l:cc}, Eq. \eqref{e:run-1}. We are left to prove that 
\begin{equation}\label{holiday}
\lim_{n\to \infty}\int_{0}^{+\infty} (W_-)_{e^*} |(\Psi_n)_{e^*}|^2 dx = 0 . 
\end{equation}
We have that  for any $\ve>0$ and $r\geq 1$, there exists $R>0$ (independent of $n$) such that
 \begin{equation*}
\int_{R}^{+\infty} (W_-)_{e^*} |(\Psi_n)_{e^*}|^2 dx \leq \|(W_-)_{e^*}\|_{L^r(R,+\infty)} \|\Psi_n\|_{2r/(r-1)}^2 \leq \ve.
\end{equation*}
 For such $R$, there exists $n_0$ such that for all $n>n_0$ one has 
\begin{equation*}
\int_{0}^{R} (W_-)_{e^*} |(\Psi_n)_{e^*}|^2 dx \leq  \|W_-\|_r \|(\Psi_n)_{e^*}\|_{L^{2r/(r-1)}(0,R)}^2 \leq \ve  
\end{equation*}
by Eq.  \eqref{e:run-1}, from which the second limit in \eqref{limit}. 

Recalling that, by Lem. \ref{l:cc} - Eq. \eqref{e:run-1}, one has $\lim_{n\to\infty}\|(\Psi_n)_e\|_{L^{2\mu+2}(I_e)} =0$ for all $e\neq e^*$, and by using Eq. \eqref{limit}, we infer 
 \begin{equation*}
 \lim_{n\to \infty} E[\Psi_n]  \geq \lim_{n\to\infty } \left( \|(\Psi_n)_{e^*}'\|^2_{L^2(\RE_+)} -\frac{1}{\mu+1} \|(\Psi_n)_{e^*}\|^{2\mu+2}_{L^{2\mu+2}(\RE_+)}\right)
 \end{equation*}
 Let  $\chi:\RE_+ \to [0,1]$ be a  function such that $\chi \in C^\infty(\RE_+)$, $\chi(0) = 0$ and $\chi(x)=1$ for all $x\geq 1$ and define 
\[\psi_n^*(x) :=  \chi(x)(\Psi_n)_{e^*}(x) , \]
so that $\psi_n^*(0)= 0$, and ${\|\psi_n^*}'\|_{L^2(\RE_+)}^2 \leq c$. We have the following inequalities (we refer to  \cite{cacciapuoti-finco-noja-rxv16} for the details)  
\begin{align}
-\nu_\mu(m) = &\lim_{n\to \infty} E[\Psi_n]  \geq 
\lim_{n\to\infty }\left( \|{\psi_n^*}'\|^2_{L^2(\RE_+)} -\frac{1}{\mu+1} \|\psi_n^*\|^{2\mu+2}_{L^{2\mu+2}}\right) \nonumber\\
 \geq &\inf \left\{E_\RE[\psi]\,\Big|\; \psi\in H^1(\RE), \, \| \psi\|_{L^2(\RE)}^2 =m  \right\} = - t_\mu(m), \label{infsol}
\end{align}
see Eqs. \eqref{thres-inf} and \eqref{nuRE} for the explicit value of $t_\mu(m)$. 

 By the bounds \eqref{nu-cond1} and  \eqref{infsol} we infer that, if the minimizing sequence $\Psi_n$ is runaway, then it must be 
\begin{equation}\label{chain}0>-E_0 m \geq -\nu_\mu(m) \geq -t_\mu(m).\end{equation}
 Next we distinguish the subcritical and the critical case. If $0<\mu<2$, it is well known (see, e.g, \cite{cazenave03}) that the infimum $ -t_\mu(m)$ is negative, finite, and that it is indeed attained for any $m>0$, moreover $t_\mu(m)$  is given by Eq. \eqref{thres-inf}, 
with
\begin{equation}\label{gammamu}
\gamma_\mu= \frac{2-\mu}{2+\mu}\left( 2\frac{(\mu+1)^{\frac 1 \mu }  }{\mu}  \int_0^1 (1-t^2)^{\frac 1 \mu -1} dt\right)^{-\frac{2\mu}{2-\mu}}\qquad 0<\mu<2.
 \end{equation} 
We conclude that, whenever $-E_0 m  < -\gamma_\mu m^{1+\frac{2\mu}{2-\mu}}$ (i.e., $m <\left(E_0/ \gamma_\mu\right)^{\frac1\mu-\frac12}$) 
there is a contradiction with the chain of inequalities \eqref{chain}, hence the minimizing sequence cannot be runaway and must converge to a certain function $\hat\Psi$. \\
If $\mu = 2$ the infimum $-t_2(m)$ exhibits a critical mass, see Eq. \eqref{nuRE}. Moreover the infimum is attained only at the critical mass  $m = \pi \sqrt 3/2 $. Since $ -t_\mu(m) =0$,  for  $m \leq \pi \sqrt 3/2$, contradicts the chain of inequalities  \eqref{chain} we conclude that also in this case the minimizing sequence cannot be runaway and must converge to a certain function $\hat\Psi$. The latter argument, together with the constraint $m<m_2^*$ (needed for the bound \eqref{world}) and the fact that by \cite[Prop. 2.3]{adami-serra-tilli-cmp16}, $(2/\pi)^{1/3}= K_{6,2}(\RE) \leq K_{6,2}(\GG)$, tell us that the value of the threshold mass is $m_2^*$. 

By Lem. \ref{l:cc} we conclude that for all $0<\mu\leq 2$ and $m$ small enough there exists a state $\hat\Psi \in H^1(\GG)$ such that  minimizing sequences converge, up to taking  subsequences, to $\hat\Psi $ in $L^p$ for $p  \geq 2$. To establish the convergence of $\Psi_n \to \hat \Psi$ in $H^1(\GG)$, and conclude the proof of the theorem,  one can repeat  the general argument used  in \cite{cacciapuoti-finco-noja-rxv16}.
\end{proof}

\appendix
\section{\label{s:appendix}An upper bound for $-\nu_\mu(m)$}
In the following proposition we prove that in the subcritical case, if the potential $W$ decays at infinity,  the infimum in \eqref{minimization} cannot exceed the infimum $-t_\mu(m)$.
\begin{proposition}\label{p:necessary}Let $0<\mu<2$. Let Assumption \ref{a:G}  hold true and assume that $W\in L^r(\GG)$ for some $r\in[1,+\infty)$. Then 
\begin{equation*}
-\nu_\mu(m) \leq -t_\mu(m).
\end{equation*}
\end{proposition}
\begin{proof}To prove the claim it is enough to exhibit a sequence $\Phi_n$ such that: $\Phi_n\in H^1(\GG)$, $\|\Phi_n\|^2\to m$, and 
\[E[\Phi_n] \to  -t_\mu(m) \] 
as $n\to \infty$. 

We recall that (see, e.g., \cite{sulem-sulem07} and \cite{cazenave03}) for $0<\mu<2$ and $m>0$, the infimum  $-t_\mu(m)$ is attained by the function 
\[
\phi_{\omega}(x) = [(\mu+1)\omega]^{\frac1{2\mu}} \sech^{\frac1\mu} (\mu \sqrt\omega x) \]
with
\[
\omega = \left(\frac{m \mu}{2 (\mu + 1)^\frac{1}{\mu}  I(\mu)}  \right)^{\frac{2\mu}{2-\mu}} , 
\]
where $I(\mu) = \int_0^1(1-t^2)^{\frac1\mu -1} dt $. Fix $e^*\in \Ed^{ex}$ and let $\Phi_n $ be defined as  
\[(\Phi_n)_e = \left\{ \begin{aligned}
&\chi \phi_\omega(\cdot-n) &e=e^*\\
& 0 \qquad & e\neq e^*  
\end{aligned}\right.\] where $\chi$ is a $C^{\infty}(\RE_+)$ function such that $\chi \in C^\infty(\RE_+)$, $\chi(0) = 0$ and $\chi(x)=1$ for all $x\geq 1$. One has 
\[\|\Phi_n\|^2 = \|\chi \phi_\omega(\cdot - n)\|^2_{L^2(\RE_+)}\to m.\] 
Moreover 
\[\begin{aligned}
E[\Phi_n] 
= &  \|(\chi \phi_\omega(\cdot - n))'\|^2_{L^2(\RE_+)}  \\ 
&   + (\chi \phi_\omega(\cdot - n), W_e^* \chi \phi_\omega(\cdot - n))_{L^2(\RE_+)} - \frac{\|\chi \phi_\omega(\cdot - n))\|^{2\mu+2}_{L^{2\mu+2}(\RE_+)}}{\mu+1} \\ 
\to&  E_\RE[\phi_\omega] = -t_\mu(m).
\end{aligned}
\]
The latter claim is an immediate consequence of the the fact that 
\[\|\phi'_\omega(\cdot - n)\|^2_{L^2(-\infty,1)} \to 0 \qquad \text{and} \qquad \|\phi_\omega(\cdot - n)\|^p_{L^p(-\infty,R)} \to 0\]  
for all $p\in[1,\infty]$ and $R>0$, as $n\to\infty$ (see also  Eq. \eqref{holiday}). 
\end{proof}


\end{document}